\documentclass[11pt,a4paper]{article}
\usepackage{amsmath,amssymb,amsthm,graphicx}
\usepackage[english]{babel}
\usepackage[pdfpagemode=UseOutlines,bookmarksnumbered=true,,bookmarksopen=true,
 pdfhighlight=/P,colorlinks=true,linkcolor=blue,citecolor=blue,urlcolor=blue,
 unicode=true]{hyperref}
\usepackage[width=0.9\textwidth,indention=0mm,justification=centerlast,
 font=small,labelfont=bf,labelsep=period]{caption}

\theoremstyle{plain}
\newtheorem{theorem}{Theorem}
\theoremstyle{remark}

\newtheorem{example}{Example}

\def\plot#1#2{\includegraphics[width=#1]{#2.pdf}}
\def\iu{{\rm i}}

\title{Differential-difference equations associated with the
fractional Lax operators}

\author{V.E.~Adler\thanks{L.D.~Landau Institute for Theoretical Physics, 1A
Ak.~Semenov, Chernogolovka 142432, Russia. E-mail: adler@itp.ac.ru} ,\quad
V.V.~Postnikov\thanks{Sochi Branch of Peoples' Friendship University of Russia,
32 Kuibyshev str., 354000 Sochi, Russia. E-mail: postnikofvv@mail.ru}}

\date{March 28, 2011}
\date{July 12, 2011}

\begin{document}\maketitle

\noindent\hrulefill
\par\noindent{\bf Abstract.}\quad
We study integrable hierarchies associated with spectral problems of the form
$P\psi=\lambda Q\psi$ where $P,Q$ are difference operators. The corresponding
nonlinear differential-difference equations can be viewed as inhomogeneous
generalizations of the Bogoyavlensky type lattices. While the latter turn into
the Korteweg--de Vries equation under the continuous limit, the lattices under
consideration provide discrete analogs of the Sawada--Kotera and
Kaup--Kupershmidt equations. The $r$-matrix formulation and several simplest
explicit solutions are presented.
\medskip

\medskip

\par
\noindent\hrulefill

\section{Introduction}

The simplest example studied in this paper is the lattice equation
\begin{equation}\label{dSK}
 u_{,t}=u^2(u_2u_1-u_{-1}u_{-2})-u(u_1-u_{-1})
\end{equation}
where we use the shorthand notations
\[
 u=u(n,t),\quad u_{,t}=\partial_t(u),\quad u_j=u(n+j,t).
\]
For the first time, this equation was derived by Tsujimoto and Hirota
\cite[eq.\,(4.12)]{Tsujimoto_Hirota} as the continuous limit of the reduced
discrete BKP hierarchy. Recall that both equations
\begin{equation}\label{VB}
 u_{,t'}=u(u_1-u_{-1})\qquad \text{and}\qquad u_{,t''}=u^2(u_2u_1-u_{-1}u_{-2})
\end{equation}
are very well known integrable models: respectively, the Volterra lattice
\cite{Zakharov_Musher_Rubenchik,Manakov} and the modified
Narita--Itoh--Bogoyavlensky lattice of the second order
\cite{Narita,Itoh,Bogoyavlensky}. One can easily verify that the flows
$\partial_{t'}$ and $\partial_{t''}$ do not commute, that is, these equations
belong to the different hierarchies. Hence, one should not expect a priori that
their linear combination remains integrable. Nevertheless, this is the case: we
will show that equation (\ref{dSK}) admits the Lax representation
\[
 L_{,t}=[A,L]
\]
with the operator $L$ equal to a ratio of two difference operators, namely,
$L=(T^2+u)^{-1}(uT^2+1)T$ where $T$ denotes the shift operator $u_k\to
u_{k+1}$.

Equation (\ref{dSK}) can be cast into the Hirota's bilinear form which admits a
family of generalizations depending on a pair of integer parameters $(l,m)$.
These generalizations were discovered by Hu, Clarkson and Bullough
\cite[eq.\,(4)]{Hu_Clarkson_Bullough} who searched for bilinear equations
admitting $N$-soliton solutions. One of the goals of our paper is to
demonstrate that this family of equations is associated with the fractional Lax
operators of the form
\begin{equation}\label{Lml}
 L=(T^m+u)^{-1}(uT^m+1)T^l.
\end{equation}
As usually, any such $L$ is associated with a whole commutative hierarchy of
equations corresponding to the sequence of difference operators $A$ of
increasing order. We denote this hierarchy dSK$^{(l,m)}$, since it can be
viewed as a discretization of the hierarchy containing the Sawada--Kotera
equation \cite{Sawada_Kotera, Caudrey_Dodd_Gibbon}
\begin{equation}\tag{SK}
 U_{,\tau}=U_5+5UU_3+5U_1U_2+5U^2U_1
\end{equation}
where we denote
\[
 U=U(x,\tau),\quad U_{,\tau}=\partial_\tau(U),\quad U_j=\partial^j_x(U).
\]
For instance, equation (\ref{dSK}) belongs to dSK$^{(1,2)}$. The concrete
formula of the continuous limit in this example is the following, at
$\varepsilon\to0$:
\begin{equation}\label{12_lim}
 u(n,t)=\frac{1}{3}+\frac{\varepsilon^2}{9}U\Bigl(x-\frac{4}{9}\varepsilon t,
    \tau+\frac{2\varepsilon^5}{135}t\Bigr),\quad x=\varepsilon n
\end{equation}
and an analogous formula exists for any $(l,m)$. It should be noted that each
of equations (\ref{VB}) apart defines a discretization of the Korteweg--de
Vries (KdV) equation $U_{,t}=U_3+6UU_1$ rather than the SK one. Moreover, it is
well known that actually all Bogoyavlensky type lattices serve as
discretizations of the KdV equation or its higher symmetries, so that an
infinite family of discrete hierarchies correspond to just one continuous.
Quite analogously, the whole family of dSK$^{(l,m)}$ hierarchies serve as
discrete analogs of the SK hierarchy. We hope that this observation makes clear
the place of these equations in the big picture of integrable systems.

On the other hand, the differential and difference cases are not quite
parallel. First, Lax operator for the SK equation
\[
 L=D^3+UD=(D-f)(D+f)D
\]
is not fractional. Lax operators given by the ratio of differential operators
were studied by Krichever \cite{Krichever}, however it seems that these
examples and (\ref{Lml}) are unrelated.

Second, let us consider the problem of discretization for another important
example, the Kaup--Kupershmidt equation \cite{Kaup, Kupershmidt,
Drinfeld_Sokolov}
\begin{equation}\tag{KK}
 U_{,\tau}=U_5+5UU_3+\frac{25}{2}U_1U_2+5U^2U_1.
\end{equation}
Recall that it is associated with the operator
\[
 L=D^3+UD+\frac{1}{2}U_{,x}=(D+f)D(D-f)
\]
and both SK and KK equations are connected through the Miura substitutions
obtained by factorization of Lax operators \cite{Sokolov_Shabat,
Fordy_Gibbons}:
\[
 U_{\rm SK}=f_{,x}-f^2,\qquad U_{\rm KK}=-2f_{,x}-f^2.
\]
Despite of this close relation, it was noted that some properties of the SK and
KK equations are rather different, see e.g. \cite{Musette_Verhoeven}. It seems
that distinctions between the lattice analogs of these equations are even more
deep. A discretization of the KK equation is presented in section \ref{s:dKK},
however, we were able to find just one operator $L$ in this case comparing to
infinite family (\ref{Lml}) in the SK case, and no discrete analog of Miura
type substitution between dSK and dKK is known.

The contents of the paper is the following. Section \ref{s:pre} contains some
necessary information on the lattices of Bogoyavlensky type, see also books
\cite{Bogoyavlensky_b, Suris}. Section \ref{s:dSK} devoted to discretization of
the SK equation contains the main results of the paper. A general construction
of the Lax pairs with operator (\ref{Lml}) is given in section \ref{s:Lax}. In
section \ref{s:r}, the $r$-matrix approach in the difference setting
\cite{Reyman_Semenov, Blaszak_Marciniak, Suris} is used to obtain explicit
formulas for the operator $A$ and to prove the commutativity of the
dSK$^{(l,m)}$ hierarchy. The continuous limit, the bilinear representation, the
simplest breather type solutions are presented in sections \ref{s:eps},
\ref{s:bilin}. Section \ref{s:dKK} is devoted to discretization of the KK
equation and section \ref{s:generic} contains several examples of coupled
lattice equations associated with more general fractional Lax operators.

\section{Preliminaries}\label{s:pre}

\subsection{Definitions and notations}

We consider differential-difference (lattice) equations of the evolutionary
form
\begin{equation}\label{ut}
 u_{,t}=f(u_m,\dots,u_{-m}),\quad
 u=u(n,t),\quad u_{,t}=\partial_t(u),\quad u_j=u(n+j,t).
\end{equation}
Such equations can be viewed as discrete analogs of continuous evolutionary
equations like KdV or SK
\[
 U_{,\tau}=F(U_k,\dots,U),\quad
 U=U(x,\tau),\quad U_{,\tau}=\partial_\tau(U),\quad U_j=\partial^j_x(U)
\]
(the orders $m$ and $k$ may not coincide under the continuous limit). The shift
operator $T:u_j\mapsto u_{j+1}$ plays the same role for equations (\ref{ut}) as
the total $x$-derivative $D:U_j\mapsto U_{j+1}$ plays in the continuous case.
Differential operators are polynomials with respect to $D$, with the
multiplication defined by the Leibniz rule $DA=D(A)+AD$ and the conjugation
defined by the rule $D^\dag=-D$. In contrast, difference operators are in
general Laurent polynomials, that is contain powers of both $T$ and $T^{-1}$,
and the rules for the multiplication and the conjugation are $TA=T(A)T$ and
$T^\dag=T^{-1}$. For short, we will use subscripts also for denoting action of
$T$ on operators, $A_j=T^j(A)$.

A lattice equation
\[
 u_{,t'}=g(u_k,\dots,u_{-k})
\]
is called symmetry of (\ref{ut}) if the compatibility condition
$D_{,t}(g)=D_{,t'}(f)$ is fulfilled, that is
\begin{equation}\label{f_comm}
 [f,g]_*:=\sum^m_{s=-m}\partial_{u_s}(f)T^s(g)
       -\sum^k_{s=-k}\partial_{u_s}(g)T^s(f)=0.
\end{equation}
The lattice is called integrable if it admits an infinite sequence of
symmetries with the order $k$ greater than any fixed number. The linear space
of all symmetries is called hierarchy. A conservation law is a relation of the
form
\[
 D_{,t}(\rho(u_k,\dots,u))=(T-1)(\sigma(u_{k+m-1},\dots,u_{-m}))
\]
which holds true in virtue of equation (\ref{ut}). The discussion of these
notions and applications to the problem of classification of integrable lattice
equations can be found in the review article by Yamilov \cite{Yamilov}.

\subsection{Bogoyavlensky lattices}

Understanding the structure of dSK$^{(l,m)}$ hierarchy is not possible without
understanding the homogeneous hierarchies of Bogoyavlensky type. A general
pattern of (local) equations from dSK$^{(l,m)}$ is given by the formula
\[
 u_{,t_k}=F^{(L+KM)}+\dots+F^{(L+M)}+F^{(L)}
\]
where $F^{(s)}$ denotes a homogeneous polynomial of degree $s$ with respect to
the variables $u_j$ and $K,L,M$ are related somehow with the parameters $l,m$
and the order $k$ of the flow. Moreover, the first and the last terms in the
sum always correspond to some (modified) lattices of Bogoyavlensky type
belonging to the different hierarchies.

This structure is explained by the following arguments, starting from the Lax
representation with the operator $L$ (\ref{Lml}). Let us consider the scaling
$u\to\delta^{-m}u$, $T\to\delta T$, then it is easy to see that the limit
$\delta\to\infty$ sends $L$ to the operator $L'=u_{-m}T^l+T^{l-m}$ and the
limit $\delta\to0$ leads to $L''=T^{m+l}+u^{-1}T^l$. Each of these operators
corresponds to its own hierarchy of homogeneous lattice equations. The total
inhomogeneous equation contains both of them together with the intermediate
terms which are necessary for preserving commutativity of the flows.

Let us consider the concrete example. One can check that the lattice
\begin{multline}\label{dSK'}
 u_{,t'}=u\bigl(w_1(w_3+w_2+w_1+w)-w_{-1}(w+w_{-1}+w_{-2}+w_{-3})\\
   -u_1(w_3+w_{-1})+u_{-1}(w_1+w_{-3})\bigr),\quad w:=u(1-u_1u_{-1})\quad
\end{multline}
is a higher symmetry of equation (\ref{dSK}). Collecting the homogeneous terms
yields
\[
 u_{,t}=F^{(4)}+F^{(2)},\quad u_{,t'}=G^{(7)}+G^{(5)}+G^{(3)}
\]
and the consistency condition of the flows splits to relations
\begin{gather*}
 [F^{(4)},G^{(7)}]_*=0,\quad [F^{(4)},G^{(5)}]_*+[F^{(2)},G^{(7)}]_*=0,\\
 [F^{(4)},G^{(3)}]_*+[F^{(2)},G^{(5)}]_*=0,\quad [F^{(2)},G^{(3)}]_*=0
\end{gather*}
where commutator $[,]_*$ is defined by equation (\ref{f_comm}). As it was
already said in Introduction, polynomials $F^{(4)}$ and $F^{(2)}$ correspond to
the modified Bogoyavlensky and Volterra lattices. Polynomials $G^{(7)}$ and
$G^{(3)}$ correspond to their symmetries and the intermediate polynomial
$G^{(5)}$ compensates inconsistency of the hierarchies.

The Bogoyavlensky hierarchy B$^{(m)}$ is associated with the operator
$L=T+uT^{-m}$ and we recall here several basic formulas regarding this case. A
detailed theory can be found in the books \cite{Bogoyavlensky, Suris}. More
general operators of the form $L=T^l+uT^{-m}$ were considered recently in the
paper \cite{Svinin}.

The simplest equation from the B$^{(m)}$ hierarchy reads
\begin{equation}\label{Bm-1}
 u_t=u(u_m+\dots+u_1-u_{-1}-\dots-u_{-m}).
\end{equation}
This equations and its higher symmetries are associated with the difference
spectral problem
\[
 \psi_1+u\psi_{-m}=\lambda\psi
\]
and admit the Lax representations
\begin{equation}\label{Bm-k}
 L_{,t_k}=[A^{(k)},L],\quad L=T+uT^{-m},\quad
 A^{(k)}=\pi_+\bigl(L^{(m+1)k}\bigr)
\end{equation}
where $\pi_+$ denotes the projection of any formal series
$A=\sum_{j<\infty}a^{(j)}T^j$ onto the linear space of polynomials with respect
to $T$:
\[
 \pi_+(A)=\sum_{0\le j<\infty}a^{(j)}T^j,\quad \pi_-(A)=\sum_{j<0}a^{(j)}T^j.
\]
In particular,
\[
 A^{(1)}=T^{m+1}+v,\quad v:=u_m+\dots+u
\]
and equation (\ref{Bm-k}) at $k=1$ is equivalent to lattice (\ref{Bm-1}). The
check is easy:
\begin{multline}\label{Bm-1_proof}
 L_{,t}-[A^{(1)},L]= u_{,t}T^{-m}-[T^{m+1}+v,T+uT^{-m}]\\
  = u_{,t}T^{-m}-(u_{m+1}-u+v-v_1)T-u(v-v_{-m})T^{-m},\quad
\end{multline}
the terms with $T$ cancel and the rest yields the equation.

In order to prove that equation (\ref{Bm-k}) correctly defines the lattice for
any $k$, we have to check that all powers of $T$ except for $T^{-m}$ vanish in
the commutator $[A^{(k)},L]$. Since $L^{m+1}$ is a Laurent polynomial with
respect to $T^{m+1}$, hence $A^{(k)}$ is a polynomial with respect to
$T^{m+1}$. Therefore the commutator contains only powers of the form
$T^{(m+1)j+1}$, $j\ge-1$. On the other hand,
\[
 [A^{(k)},L]=-[\pi_-\bigl(L^{(m+1)k}\bigr),L],
\]
so that the commutator does not contain positive powers of $T$ and only one
possible power $T^{-m}$ remains.

It can be proven that equations (\ref{Bm-k}) define a special reduction in the
Lax pair with a generic operator
$L=T+u^{(0)}+u^{(1)}T^{-1}+\dots+u^{(m)}T^{-m}$. In this case one can choose
operators $A$ in the form $A=\pi_+(L^k)$ with arbitrary $k$. For instance, the
Toda lattice hierarchy appears at $m=1$. This type of multi-field systems was
studied, for instance, in papers \cite{Blaszak_Marciniak, Svinin}.

\section{Discretizations of the Sawada--Kotera equation}\label{s:dSK}

\subsection{Lax representation}\label{s:Lax}

Let us consider the difference spectral problem
\begin{equation}\label{psi}
 u\psi_{m+l}+\psi_l=\lambda(\psi_m+u\psi)
\end{equation}
where $m,l$ are integers. We assume that $m,l$ are positive and coprime,
without loss of generality, since the general case can be obtained by
refinement of the mesh and/or change of its directions. It is less obvious that
the numbers $m$ and $l$ can be exchanged: spectral problem (\ref{psi}) is
equivalent to
\[
 u\varphi_{m+l}+\varphi_m=\mu(\varphi_l+u\varphi)
\]
under the change
\begin{equation}\label{lm}
 \psi(n)=\varkappa^n\varphi(n),\quad \lambda=-\varkappa^l,\quad \mu=-\varkappa^{-m}.
\end{equation}

In the operator form, equation (\ref{psi}) reads
\begin{equation}\label{PQ}
 P\psi=\lambda Q\psi,\quad P=(uT^m+1)T^l,\quad Q=T^m+u.
\end{equation}
The isospectral deformations are defined by equation $\psi_{,t}=A\psi$ with
some difference operator $A$. The corresponding Lax equation
\begin{equation}\label{Lax}
 L_{,t}=[A,L],\quad L=Q^{-1}P
\end{equation}
can be rewritten as the system
\begin{equation}\label{PQt}
 P_{,t}=BP-PA,\quad Q_{,t}=BQ-QA
\end{equation}
where one of equations can be considered just as a definition of $B$. Let $P,Q$
be as in (\ref{PQ}), then this system is equivalent to equations
\begin{gather}
\nonumber
 u_{,t}=B(T^m+u)-(T^m+u)A,\\
\label{utAB}
 B(T^{2m}-1)=A_mT^{2m}-A_l+uAT^m-uA_{m+l}T^m.
\end{gather}
In order to resolve the latter we make the assumption that operator $A$ is of
the form
\begin{equation}\label{AF}
 A=F(T^m-T^{-m})
\end{equation}
then $B$ is found as the difference operator
\begin{equation}\label{ABF}
 B=F_mT^m-F_1T^{-m}+u(F-F_{m+1})
\end{equation}
while first equation (\ref{utAB}) turns into
\begin{equation}\label{utF}
 u_{,t}=T^mFu+uFT^{-m}-uT^mF_l-F_lT^{-m}u+F_m-F_l+u(F-F_{m+l})u.
\end{equation}
It is clear that the same evolution of the variable $u$ is defined by the
conjugated operator $F^\dag$ and, moreover, all terms $T^j$, $j\nmid m$ can be
thrown away. This means that we can find $F$ as a self-adjoint operator
$F=F^\dag$ which is a Laurent polynomial with respect to the powers $T^m$:
\begin{equation}\label{F}
 F=f^{(k)}T^{km}+\dots+f^{(1)}T^m+f^{(0)}+T^{-m}f^{(1)}+\dots+T^{-km}f^{(k)},
 \quad k\ge0.
\end{equation}
Certainly, the coefficients depend on $k,l,m$, so that it would be more
rigorous to write $f^{(j,k,l,m)}$ instead of $f^{(j)}$, but we will consider
these numbers fixed at the moment.

Collecting the coefficients at $T^{jm}$, $j>0$, yields the relations
\begin{multline}\label{rec}\qquad
 u_{jm}f^{(j-1)}_m-uf^{(j-1)}_{m+l}
  = f^{(j)}_l-f^{(j)}_m+uu_{jm}(f^{(j)}_{m+l}-f^{(j)}) \\
   +u_{jm}f^{(j+1)}_l-uf^{(j+1)},\quad j=1,\dots,k+1,\qquad
\end{multline}
where it is assumed for convenience that $f^{(j)}=0$ at $j>k$. The coefficient
at $T^0$ gives an evolutionary equation for $u$:
\begin{equation}\label{utk}
 u_{,t}=2u(f^{(1)}-f^{(1)}_l)+u^2(f^{(0)}-f^{(0)}_{m+l})+f^{(0)}_m-f^{(0)}_l.
\end{equation}
System of equations (\ref{rec}), (\ref{utk}) defines the $k$-th flow in the
hierarchy dSK$^{(l,m)}$.

If we are interested in the local evolution only then we require that all
$f^{(j)}$ can be recurrently found as functions of an finite set of variables
$u_i$. In this case a certain restriction on the values of $k$ appears and a
part of the flows is rejected. Indeed, consider equation (\ref{rec}) at
$j=k+1$,
\begin{equation}\label{uf}
 u_{(k+1)m}f^{(k)}_m=uf^{(k)}_{m+l},
\end{equation}
or
\[
 (T^l-1)(\log f^{(k)}_m)=(T^{(k+1)m}-1)(\log u).
\]
It can be proven that it is solvable with respect to $f^{(k)}$ if and only if
$(k+1)m$ is divisible by $l$ and the solution is, up to a constant factor,
\begin{equation}\label{fuk}
 f^{(k)}=u_{-m}u_{l-m}\cdots u_{(s-1)l-m},\quad (k+1)m=sl.
\end{equation}
Since $l$ and $m$ are coprimes, hence the local flows may appear only if
$k=pl-1$ and $s=mp$. The fact that the rest equations (\ref{rec}) for such $k$
are solvable indeed will be verified later in section \ref{s:r}. The case $l=1$
is the only one when there are no restrictions on $k$ and the simplest choice
$k=0$ brings in this case to the following family of lattices.

\begin{theorem}
For any $m>0$, the simplest equation in the hierarchy dSK$^{(1,m)}$
\begin{equation}\label{utm}
 u_{,t}=u^2(u_m\cdots u_1-u_{-1}\cdots u_{-m})
      -u(u_{m-1}\cdots u_1-u_{-1}\cdots u_{1-m})
\end{equation}
possesses Lax representation (\ref{Lax}) with the operators
\begin{gather*}
 P=uT^{m+1}+T,\quad Q=T^m+u,\\
 A=f(T^{-m}-T^m),\quad B=f_1T^{-m}-f_mT^m+u(f_{m+1}-f)
\end{gather*}
where $f=u_{-1}\cdots u_{-m}$.
\end{theorem}
\begin{proof}
A direct computation (cf with (\ref{Bm-1_proof})) proves that both equations
(\ref{PQt}) with given $P,Q,A,B$ are equivalent to relations
\[
 u_mf_m=uf_{m+1},\quad u_{,t}=u^2(f_{m+1}-f)-f_m+f_1.
\]
The former defines the variable $f$ (up to a constant factor) and the latter is
equivalent to lattice (\ref{utm}).
\end{proof}

In particular, equation (\ref{utm}) at $m=2$ coincide with (\ref{dSK}) and at
$m=1$ it is just the modified Volterra lattice
\[
 u_{,t}=u^2(u_1-u_{-1}).
\]
It should be remarked that gauge equivalence (\ref{lm}) between the spectral
problems can be extended on the level of nonlinear equations and the same flow
(\ref{utm}) appears also as a member of dSK$^{(m,1)}$ hierarchy. However,
operator (\ref{F}) is much more complicated in this case: it contains all
powers $T^{m-1},T^{m-2},\dots,T^{1-m}$ comparing with just $F=f^{(0)}$ in
dSK$^{(1,m)}$ case.

Computing of higher symmetries quickly becomes involved, because finding of $F$
requires (discrete) integration of rather bulky expressions. For instance, the
second flow in the hierarchy dSK$^{(1,m)}$ is, according to (\ref{utk}), of the
form
\[
 u_{,t'}=2u(f^{(1)}-f^{(1)}_1)+u^2(f^{(0)}-f^{(0)}_{m+1})+f^{(0)}_m-f^{(0)}_1
\]
where functions $f^{(1)},f^{(0)}$ are defined by relations
\[
 u_{2m}f^{(1)}_m=uf^{(1)}_{m+1},\quad
 u_mf^{(0)}_m-uf^{(0)}_{m+1}=f^{(1)}_1-f^{(1)}_m-uu_m(f^{(1)}-f^{(1)}_{m+1}).
\]
This yields, up to integration constants,
\begin{gather*}
 f^{(1)}=u_{m-1}\cdots u_{-m},\quad
 f^{(0)}=(w+\dots+w_{-2m+1})u_{-1}\cdots u_{-m},\\
 w:=(1-u_{m-1}u_{-1})u_{m-2}\cdots u_0
\end{gather*}
(at $m=2$ equation (\ref{dSK'}) appears). One can check straightforwardly that
the obtained flow commutes with (\ref{utm}) indeed. A general proof and a way
to bypass the integration are given below in section \ref{s:r}.

Adopting nonlocal variables leads to some extension of the hierarchy. In this
case we consider equation (\ref{uf}) as a constraint which defines the variable
$f^{(k)}$ for any $k$. Then we arrive to the following system which generalizes
(\ref{utm}) for any $l$, making the picture more uniform. We will return to
this system in section \ref{s:bilin}.

\begin{theorem}
For any coprime $m,l$, the simplest system in the extended dSK$^{(l,m)}$
hierarchy
\begin{equation}\label{0lm}
 u_mf_m=uf_{m+l},\quad u_{,t}=u^2(f-f_{m+l})+f_m-f_l
\end{equation}
possesses Lax representation (\ref{Lax}) with operators
\begin{gather*}
 P=uT^{m+l}+T^l,\quad Q=T^m+u,\\
 A=f(T^{-m}-T^m),\quad B=f_lT^{-m}-f_mT^m+u(f_{m+l}-f).
\end{gather*}
\end{theorem}

\subsection{Modified lattices}

Equations under consideration can be rewritten in several ways by use of
difference substitutions. The simplest kind of substitution is introducing a
potential. Let $A$ be a constant operator, then substitution $u=A(v)$ maps
solutions of equation $v_{,t}=f[A(v)]$ into solutions of equation
$u_{,t}=A(f[u])$. Table \ref{tab:examples} contains several instances of such
kind, up to the change $u\to e^u$, $v\to e^v$.

\begin{table}[t]
\hrulefill
\begin{align*}
&m=2:     && u_{,t}=u^2(u_2u_1-u_{-1}u_{-2})-u(u_1-u_{-1})\\
&u=v_1v   && v_{,t}=v_1v^3v_{-1}(v_2v_1-v_{-1}v_{-2})-v^2(v_1-v_{-1})\\[2ex]
&m=3:     && u_{,t}=u^2(u_3u_2u_1-u_{-1}u_{-2}u_{-3})-u(u_2u_1-u_{-1}u_{-2})\\
&v=u_1u   && v_{,t}=v(v_3v_1+v_2v-vv_{-2}-v_{-1}v_{-3})-v(v_2+v_1-v_{-1}-v_{-2})\\
&u=v_2v_1v&& v_{,t}=v_2v^2_1v^4v^2_{-1}v_{-2}(v_3v_2v_1-v_{-1}v_{-2}v_{-3})\\
&         && \qquad\qquad -v_1v^3v_{-1}(v_2v_1-v_{-1}v_{-2})\\[2ex]
&m=4:     && u_{,t}=u^2(u_4u_3u_2u_1-u_{-1}u_{-2}u_{-3}u_{-4})
                       -u(u_3u_2u_1-u_{-1}u_{-2}u_{-3})\\
&u=v_2v   && v_{,t}=v_2v_1v^3v_{-1}v_{-2}(v_4v_3v_2v_1-v_{-1}v_{-2}v_{-3}v_{-4})\\
&         && \qquad\qquad -v_1v^2v_{-1}(v_3v_2v_1-v_{-1}v_{-2}v_{-3})
\end{align*}
\caption{Examples of lattices (\ref{utm}) from dSK$^{(1,m)}$ and their modifications}
\label{tab:examples}
\hrulefill
\end{table}

Another kind of substitutions are Miura type transformations. Let $\varphi$ be
a particular solution of spectral problem (\ref{psi}) corresponding to a value
$\lambda=\alpha$ of the spectral parameter. Then one readily finds that the
ratio $h=\varphi_1/\varphi$ is related with the potential $u$ by formula
\[
 M^-:\quad
 u=\frac{\alpha h_{m-1}\cdots h-h_{l-1}\cdots h}{h_{m+l-1}\cdots h-\alpha}.
\]
This defines a difference substitution, according to the following statement.

\begin{theorem}
Let $u$ satisfies an equation (\ref{utk}) from dSK$^{(l,m)}$, then
$h=\varphi_1/\varphi$ also satisfies a lattice equation which can be written as
a conservation law
\begin{equation}\label{htS}
 (\log h)_{,t}=(T-1)S[h].
\end{equation}
\end{theorem}
\begin{proof}
Since $\varphi$ is governed by equation
$\varphi_{,t}=A\varphi=F(\varphi_m-\varphi_{-m})$, hence
\[
 (\log h)_{,t}=(T-1)(\log\varphi)_{,t}
   =(T-1)\Bigl(\frac{1}{\varphi}F(\varphi_m-\varphi_{-m})\Bigr).
\]
Coefficients of the operator $F$ are functions on the variables $h_j$, being
functions on $u_j$'s. The ratios of the form $\varphi_k/\varphi$ can be
expressed through $h_j$ as well and therefore an equation of the form
(\ref{htS}) holds.
\end{proof}

It is worth noticing that an infinite sequence of conservation laws for the
original lattice (\ref{utk}) can be obtained from (\ref{htS}) by use of the
classical trick with the inversion of Miura map $u=M^-(h,\alpha)$ as a formal
power series with respect to $\alpha$ \cite{Miura_Gardner_Kruskal}.

Second Miura map is obtained by replacing $h\to1/h$, $\alpha\to1/\alpha$ which
results in the mapping
\[
 M^+:\quad
 u=\frac{\alpha h_{m+l-1}\cdots h_l-h_{m+l-1}\cdots h_m}{h_{m+l-1}\cdots h-\alpha}.
\]
This substitution relates the same equations as $M^-$, due to invariance of the
spectral problem with respect to the change $n\to-n$, $\lambda\to1/\lambda$.
Therefore, the composition $M^-(M^+)^{-1}$ defines a B\"acklund transformation
which relates two copies of the dSK$^{(l,m)}$ hierarchy. Recall that B\"acklund
transformation for the continuous SK equation was derived in
\cite{Satsuma_Kaup}.

A particular example at $l=2,m=1$ is given by substitutions
\[
 M^-:~ u=\frac{(\alpha-h_1)h}{h_2h_1h-\alpha},\quad
 M^+:~ u=\frac{h_2(\alpha-h_1)}{h_2h_1h-\alpha}
\]
which map solutions of the modified equation
\begin{gather*}
 h_{,t}=\frac{h(\alpha-h)}{h_1hh_{-1}-\alpha}\left(
  \frac{h(\alpha-h_1)(\alpha-h_{-1})(h_2h_1-h_{-1}h_{-2})}
   {(h_2h_1h-\alpha)(hh_{-1}h_{-2}-\alpha)}-h_1+h_{-1}\right)
\end{gather*}
into solutions of (\ref{dSK}).

\subsection{$r$-matrix formulation}\label{s:r}

In this section we prove that:

(i) if the constraint (\ref{uf}) is resolved by formula (\ref{fuk}) then the
further recurrent relations (\ref{rec}) are solved in the local form as well,
so that the (local) hierarchy dSK$^{(l,m)}$ is correctly defined;

(ii) the flows corresponding to the different $k$ commute.
\smallskip

In achieving this goal the $r$-matrix approach is an indispensable tool, see
e.g. \cite{Reyman_Semenov, Blaszak_Marciniak, Suris}. Let us consider the Lie
algebra of the formal Laurent series with respect to the powers $T^m$ of the
shift operator:
\[
 {\mathfrak g}^{(m)}=\Bigl\{\sum_{j<\infty}g^{(j)}T^{jm}\Bigr\}
\]
with the commutator $[A,B]=AB-BA$. It is easy to see that any element
\[
 G=g^{(k+1)}T^{(k+1)m}+g^{(k)}T^{km}+g^{(k-1)}T^{(k-1)m}+\dots
\]
of this Lie algebra admits an unique decomposition of the form
\begin{equation}\label{FGH}
 G=F(T^m-T^{-m})+H
\end{equation}
where $F=F^\dag$ is a self-conjugated difference operator and $H$ is a formal
series which contains only nonpositive powers of $T^m$. Each of the linear
spaces
\[
 {\mathfrak g}^{(m)}_+ = \bigl\{F(T^m-T^{-m})|~F=F^\dag\bigr\},\quad
 {\mathfrak g}^{(m)}_- = \Bigl\{\sum_{j\le0}h^{(j)}T^{jm}\Bigr\}
\]
constitutes a Lie algebra: for ${\mathfrak g}^{(m)}_-$ this is obvious and
for ${\mathfrak g}^{(m)}_+$ we have
\[
 [F(T^m-T^{-m}),F'(T^m-T^{-m})]=(P+P^\dag)(T^m-T^{-m})
\]
where $P=F(T^m-T^{-m})F'$.

Thus, formula (\ref{FGH}) is the decomposition (in the vector space sense)
\[
 {\mathfrak g}^{(m)}={\mathfrak g}^{(m)}_+\oplus{\mathfrak g}^{(m)}_-
\]
of the Lie algebra into the direct sum of two Lie subalgebras. This
decomposition defines the projections $\pi_\pm$ on the ${\mathfrak
g}^{(m)}_\pm$ component and the $r$-matrix $r=\frac{1}{2}(\pi_+-\pi_-)$. Now we
can formulate the following theorem about Lax equations (\ref{PQ}), (\ref{Lax})
with fractional $L$ operator.

\begin{theorem}\label{th:Ltk}
Let $l,m$ be coprime, $P=(uT^m+1)T^l$, $Q=T^m+u$ and let $L=Q^{-1}P$ be
expanded as a formal Laurent series. Then the flows
\begin{equation}\label{Ltp}
 L_{,t_p}=[\pi_+(L^{pm}),L]
\end{equation}
are correctly defined for all $p=1,2,\dots$, coincide with the dSK$^{(l,m)}$
flows introduced by equations (\ref{rec}), (\ref{utk}) and commute with each
other.
\end{theorem}
\begin{proof}
After expanding, $L$ takes the form
\begin{align*}
 L&=(1-u_{-m}T^{-m}+(u_{-m}T^{-m})^2-\dots)(u_{-m}+T^{-m})T^l\\
  &=u_{-m}T^l+(1-u_{-m}u_{-2m})T^{l-m}+\dots\,.
\end{align*}
Differentiating this series turns (\ref{Ltp}) into an infinite system of
equations for a single variable $u$, and the correctness means that all these
equations must coincide. To prove this, we compare representation (\ref{Ltp})
with Lax equation (\ref{Lax}) in fractional form.

Notice that $L$ itself does not belong to the Lie algebra ${\mathfrak
g}^{(m)}$, but its power $G=L^{pm}$ does, so that the projection
$A=\pi_+(G)=F(T^m-T^{-m})$ makes sense. We denote the order of operator $F$ as
$k=pl-1$, in agreement with (\ref{F}) and (\ref{fuk}). The coefficients of $F$
are uniquely computed from coefficients of $G$ accordingly to the recurrent
relations
\[
 f^{(k+2)}=f^{(k+1)}=0,\quad
 f^{(j)}=g^{(j+1)}+f^{(j+2)},\quad j=k,k-1,\dots,0
\]
so that all coefficients are local functions of $u_j$ (in particular, $f^{(k)}$
is given by (\ref{fuk})). Moreover, the order of (\ref{Ltp}) right hand side is
equal to $l$, because $[\pi_+(G),L]=-[\pi_-(G),L]$. This proves that $F$
provides a solution of the recurrent relations (\ref{rec}) as well (which is
unique up to integration constants). Indeed, these relations were derived from
the condition that terms with $T^{(k+1)m}$,\dots,$T^m$ in equation (\ref{utF})
cancel which is equivalent to cancellation of the powers
$T^{(k+1)m+l}$,\dots,$T^{m+l}$ in the original Lax equation (\ref{Lax}). Thus,
flow (\ref{Ltp}) coincides with a flow from dSK$^{(l,m)}$ which is, therefore,
local. On the other hand, this proves correctness of (\ref{Ltp}), since the
whole infinite set of equations turns out to be equivalent to the single
equation (\ref{utk}).

The proof of the commutativity is standard. Let $G'=L^{p'm}$ and $A'=\pi_+(G')$
then
\[
 (L_{,t_p})_{,t_{p'}}-(L_{,t_{p'}})_{,t_p}=[A_{t_{p'}}-A'_{t_p}+[A,A'],L],
\]
so it is sufficient to prove that
\[
 A_{t_{p'}}-A'_{t_p}+[A,A']=0.
\]
Since $A_{t_{p'}}=\pi_+([A',G])$ and $[G,G']=0$, this is equivalent to
\begin{align*}
 & \pi_+\bigl([A',G]-[A,G']+[A,A']\bigr)\\
 &\quad =\pi_+\bigl([G'-\pi_-(G'),G]-[G-\pi_-(G),G']+[G-\pi_-(G),G'-\pi_-(G')]\bigr)\\
 &\quad = \pi_+\bigl([\pi_-(G),\pi_-(G')]\bigr)=0
\end{align*}
as required.
\end{proof}

\subsection{Continuous limit}\label{s:eps}

Here we compute the continuous limit for the basic flow of the extended
hierarchy dSK$^{(l,m)}$ defined by equation (\ref{0lm}). There is a certain
technical difficulty in the prolongation of the continuous limit on the
variable $f$ which is not local at $l\ne1$. In order to solve the constraint,
this variable should be considered as a series with respect to the small
parameter. Up to this complication the continuous limit is very similar to
example (\ref{12_lim}) from Introduction. We postulate that, at
$\varepsilon\to0$, the variables $u,f$ are of the form
\begin{equation}\label{ufeps}
\begin{gathered}[b]
 u(n,t)=a+ab\varepsilon^2U(x+c\varepsilon t,\tau+d\varepsilon^2t),\\
 f(n,t)=1+\sum^\infty_{s=2}\varepsilon^s
   Y_s(x+c\varepsilon t,\tau+d\varepsilon^2t),\quad x=\varepsilon n
\end{gathered}
\end{equation}
with undetermined coefficients $a,b,c,d$. Functions $Y_s$ are expressed through
the function $U$ and its partial derivatives with respect to $x$ after
substituting into first equation (\ref{0lm}) and taking the Taylor expansion
about $\varepsilon=0$ (clearly, one can neglect the dependence on $t$ here). We
find, omitting the unessential integration constants:
\begin{align*}
 Y_2&=\frac{mb}{l}U,\\
 Y_3&=-\frac{m(m+l)b}{2l}U_1,\\
 Y_4&=\frac{m(m+l)(2m+l)b}{12l}U_2+\frac{m(m-l)b^2}{2l^2}U^2,\\
 Y_5&=-\frac{m^2(m+l)^2b}{24l}U_3-\frac{m(m^2-l^2)b^2}{2l^2}UU_1,\\
 Y_6&=\frac{m(m+l)(2m+l)(3m^2+3ml-l^2)b}{720l}U_4
             +\frac{m(m^2-l^2)(3m+2l)b^2}{24l^2}U^2_1 \\
     &\qquad +\frac{m(m^2-l^2)(2m+l)b^2}{12l^2}UU_2
             +\frac{m(m-l)(m-2l)b^3}{6l^3}U^3.
\end{align*}
This is enough, since we need only terms up to $\varepsilon^7$ when
substituting into second equation (\ref{0lm}). The coefficients $a,c$ are found
from the requirement that the low order terms vanish while the coefficients
$b,d$ are responsible for the scaling of $U$ and $t$ and can be chosen
arbitrarily. Finally, we come to the following statement.

\begin{theorem}
Continuous limit (\ref{ufeps}) with the values of parameters
\[
 a=\frac{m-l}{m+l},\quad b=\frac{ml}{6},\quad
 c=2m,\quad d=\frac{m^3(l^2-m^2)}{180}
\]
sends systems (\ref{0lm}) into the Sawada--Kotera equation
\[
 U_{,\tau}=U_5+5UU_3+5U_1U_2+5U^2U_1.
\]
\end{theorem}

The higher flows of the SK hierarchy can be derived analogously from suitable
linear combinations of the dSK$^{(l,m)}$ flows. However, the general formulas
become rather complicated and we restrict ourselves by the following concrete
example corresponding to the local hierarchy dSK$^{(1,2)}$. Let
$u_{,t}=88u_{,t_1}+27u_{,t_2}$ where the flows $\partial_{t_1}$ and
$\partial_{t_2}$ are defined by equations (\ref{dSK}) and (\ref{dSK'})
respectively, then the formula
\[
 u(n,t)=\frac{1}{3}+\frac{\varepsilon^2}{9}
  U\Bigl(x-\frac{200}{9}\varepsilon t,\tau-\frac{16\varepsilon^7}{189}t\Bigr),
  \quad x=\varepsilon n
\]
defines the continuous limit to the 7-th order symmetry of SK equation
\begin{multline*}
 \qquad U_{,\tau}=U_7+7UU_5+14U_1U_4+21U_2U_3+14U^2U_3\\
  +42UU_1U_2+7U^3_1+\frac{28}{3}U^3U_1.\qquad
\end{multline*}
It is well known that there are gaps in the sequence of orders $k$ of equations
from the SK hierarchy, namely, the restrictions $k\nmid2,3$ are fulfilled, so
that the next higher symmetry is of 11-th order. The natural question appears,
how this agrees with relations (\ref{F})--(\ref{utk}) or (\ref{Ltp}) which show
that in the discrete case there are no gaps multiple 3. It turns out that their
appearance is an artefact of the continuous limit. A straightforward
computation shows that if we consider a linear combination with the next
dSK$^{(1,2)}$ flow $u_{,t}=u_{,t_1}+\alpha u_{,t_2}+\beta u_{,t_3}$ and set
\[
 u(n,t)=a+b\varepsilon^2U(x+c\varepsilon t,\tau+d\varepsilon^9t),\quad
 x=\varepsilon n
\]
then all parameters are uniquely determined by the condition of vanishing the
terms up to $\varepsilon^{10}$, however then the coefficients at
$\varepsilon^{11}$ cancel automatically and only the trivial flow $U_{,\tau}=0$
appears.

\subsection{Bilinear equations}\label{s:bilin}

The constraint (\ref{uf}) can be solved by introducing additional variables and
this leads to a convenient representation of the basic system (\ref{0lm}) of
the extended dSK$^{(l,m)}$ hierarchy. Let
\[
 u=\frac{v_l}{v},\quad f=\frac{v}{v_{-m}}
\]
then first equation (\ref{0lm}) is satisfied identically and the second one is
equivalent to
\[
 (T^l-1)\frac{v_{,t}}{v}=(T^m-1)\Bigl(\frac{v}{v_{l-m}}-\frac{v_l}{v_{-m}}\Bigr).
\]
Further substitutions
\[
 v=\frac{w_m}{w}\quad\Rightarrow\quad
 u=\frac{w_{m+l}w}{w_mw_l},\quad f=\frac{w_mw_{-m}}{w^2}
\]
bring to the bilinear equation
\begin{equation}\label{wlm}
 w_{l,t}w-w_lw_{,t}=w_mw_{l-m}-w_{-m}w_{l+m}.
\end{equation}
For the first time, it appeared in paper \cite{Hu_Clarkson_Bullough}, in a
slightly more general form
\[
 w_{l,t}w-w_lw_{,t}=w_mw_{l-m}-\alpha w_{-m}w_{l+m}+\beta ww_l
\]
which is reduced to (\ref{wlm}) by the point change $\tilde w(n,t)=e^{\beta
t}\alpha^{n^2}w(n,t)$. In particular, it was proven in
\cite{Hu_Clarkson_Bullough} that this equation admits $N$-soliton solutions.
Here, we consider in more details a specification of 2-soliton formula which
leads to the breather solution.

\begin{figure}[t!]
\plot{35mm}{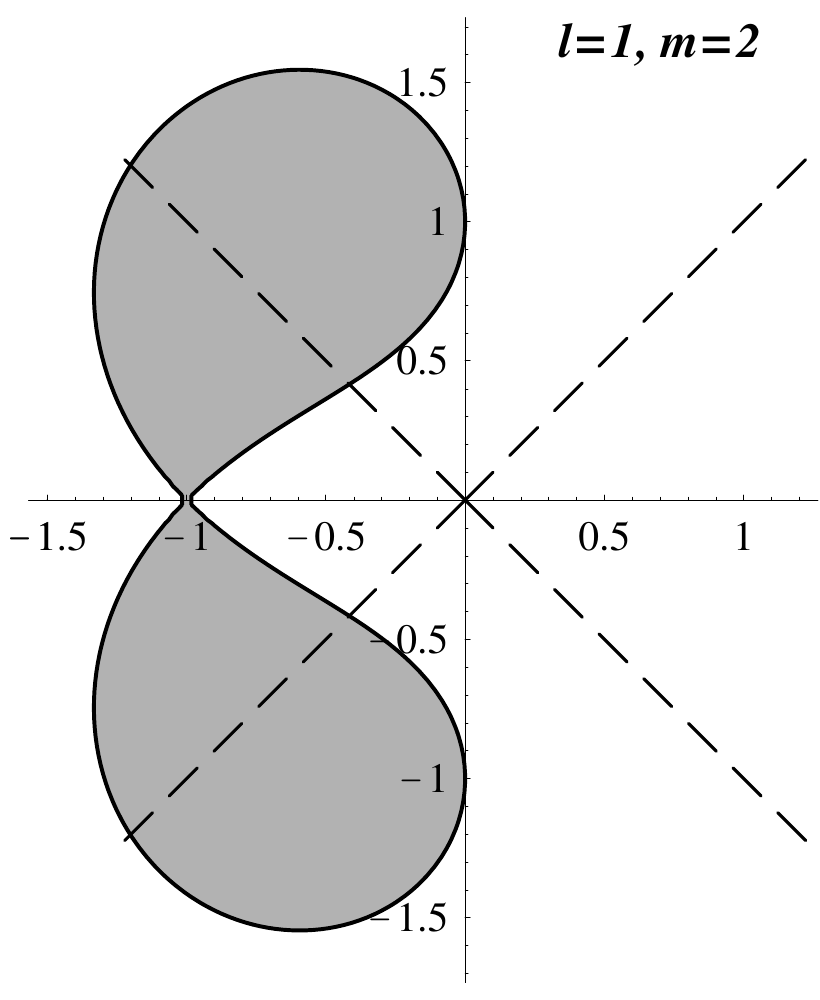}\hfill\plot{35mm}{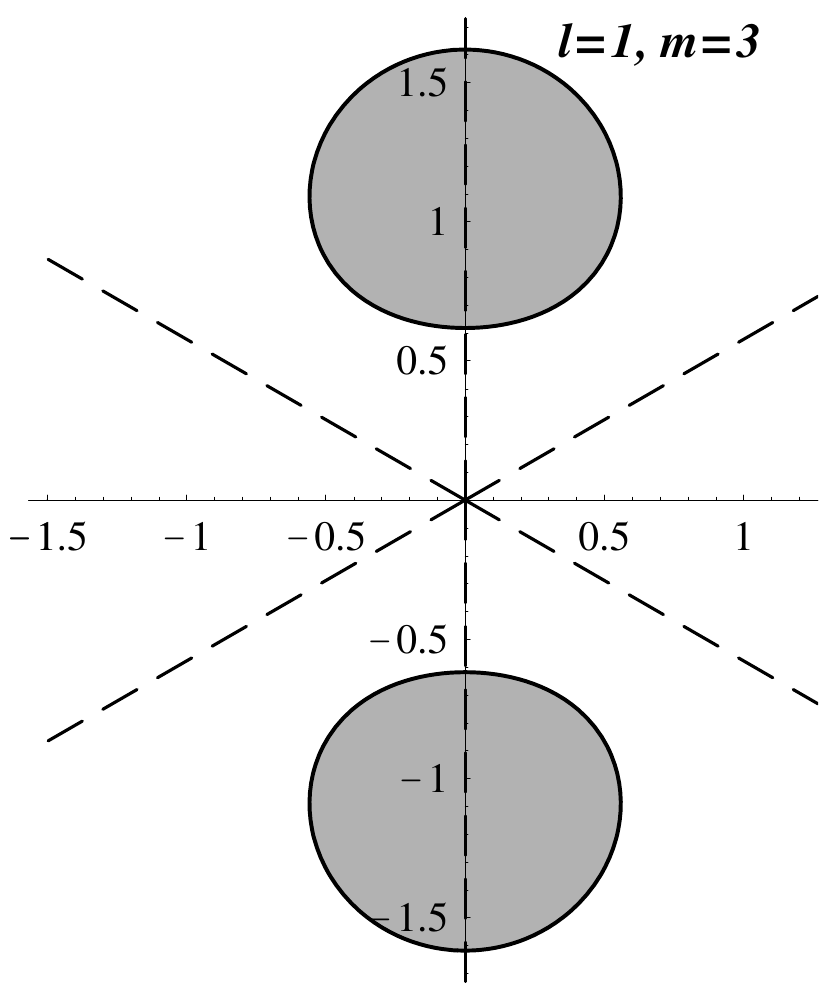}\hfill\plot{35mm}{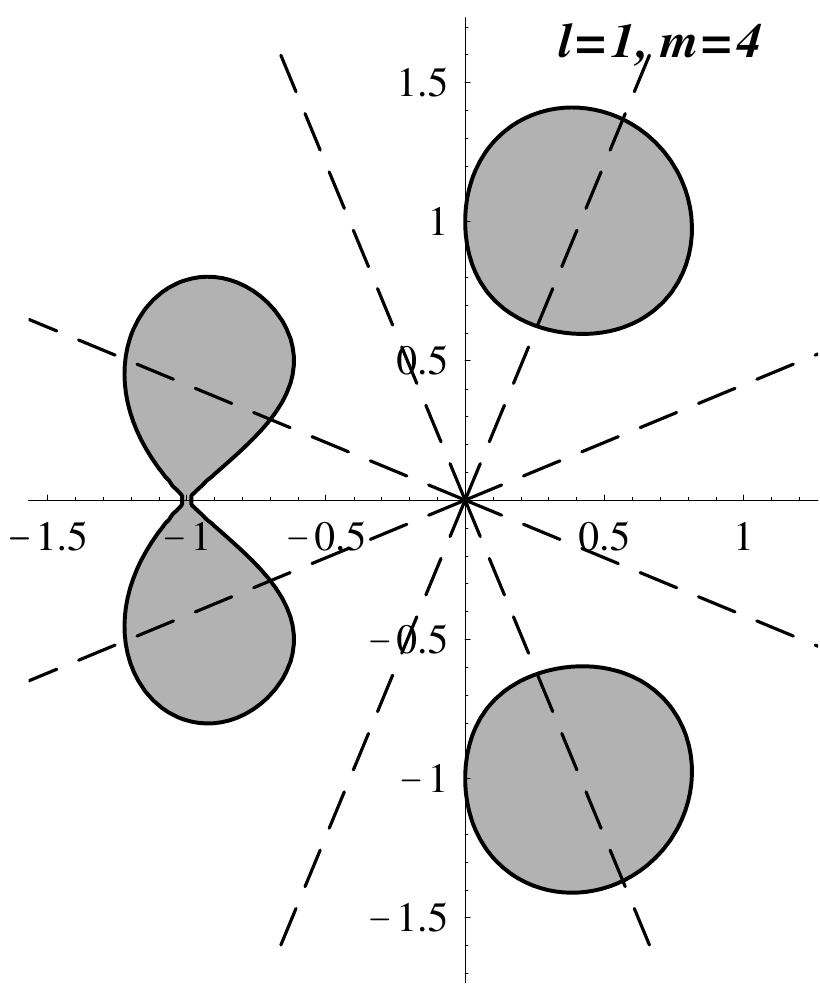}\\
\plot{35mm}{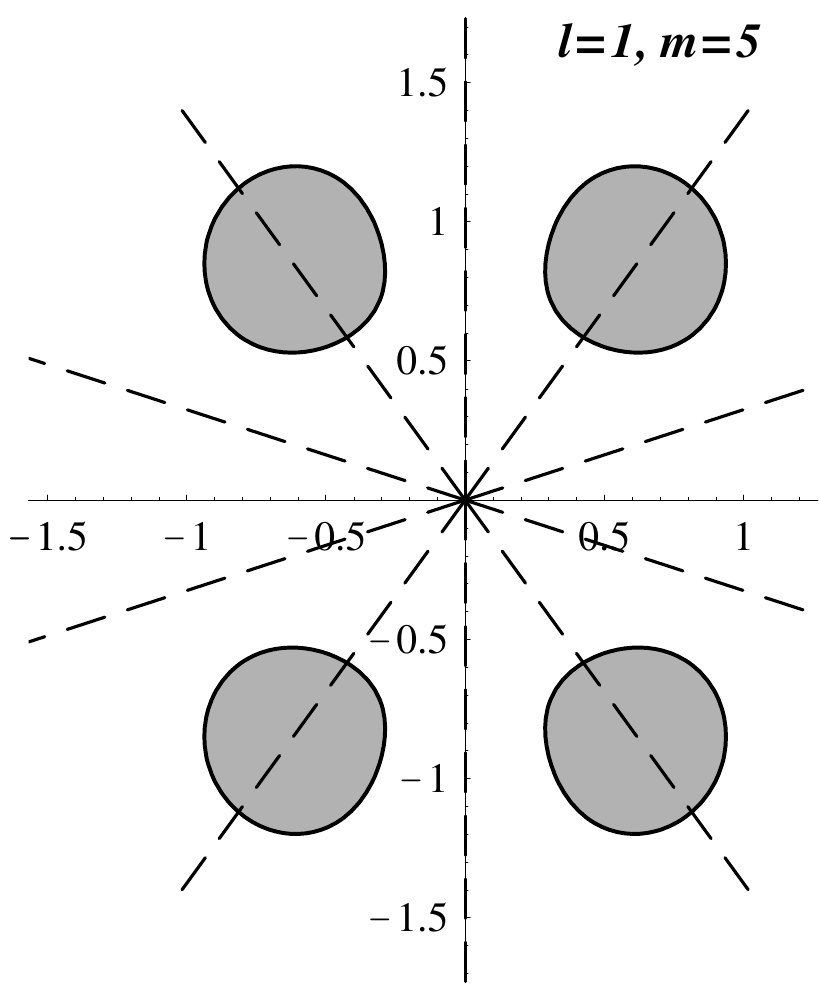}\hfill\plot{35mm}{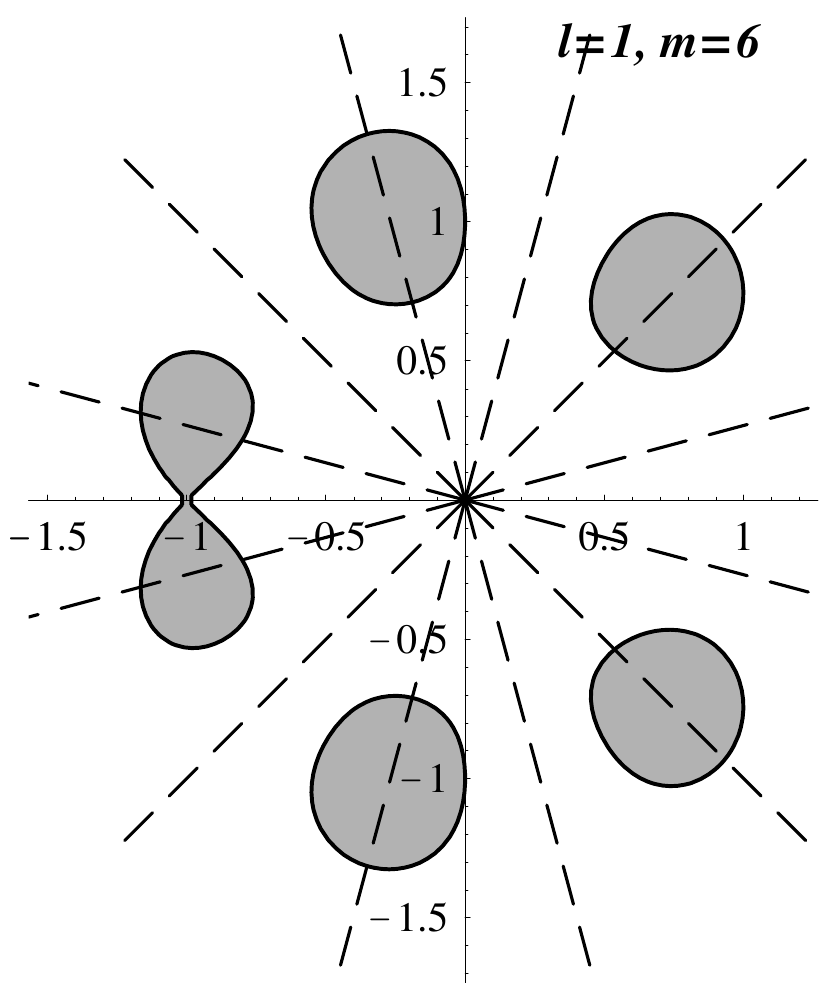}\hfill\plot{35mm}{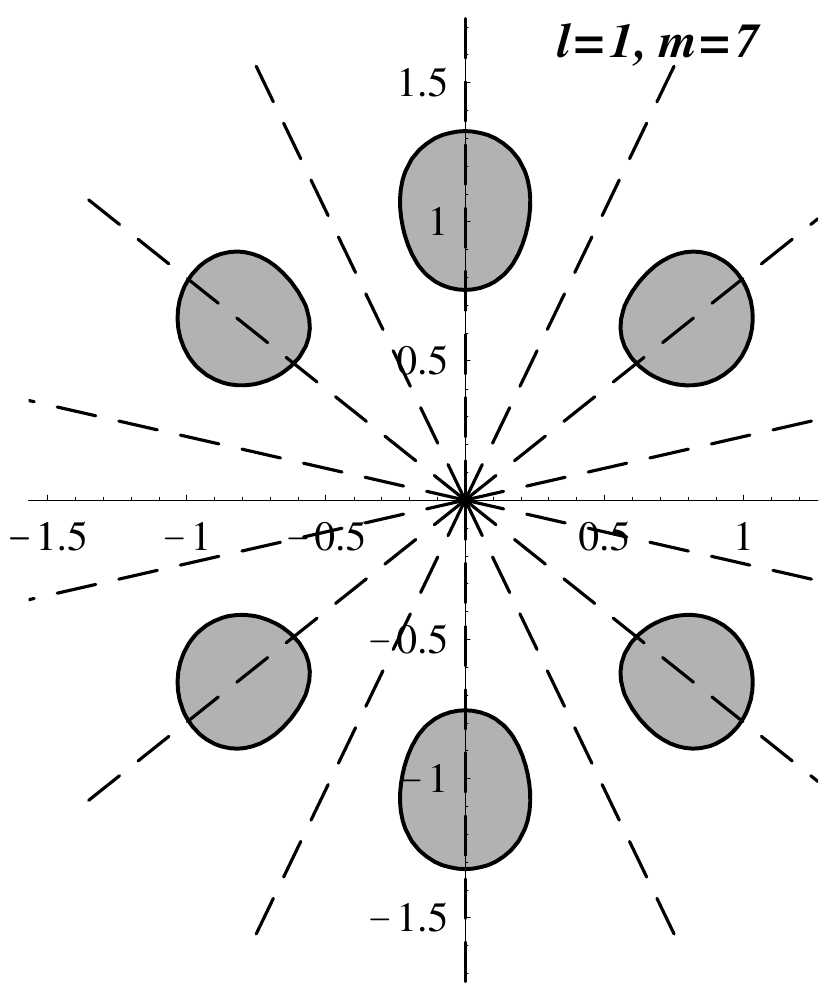}\\
\plot{35mm}{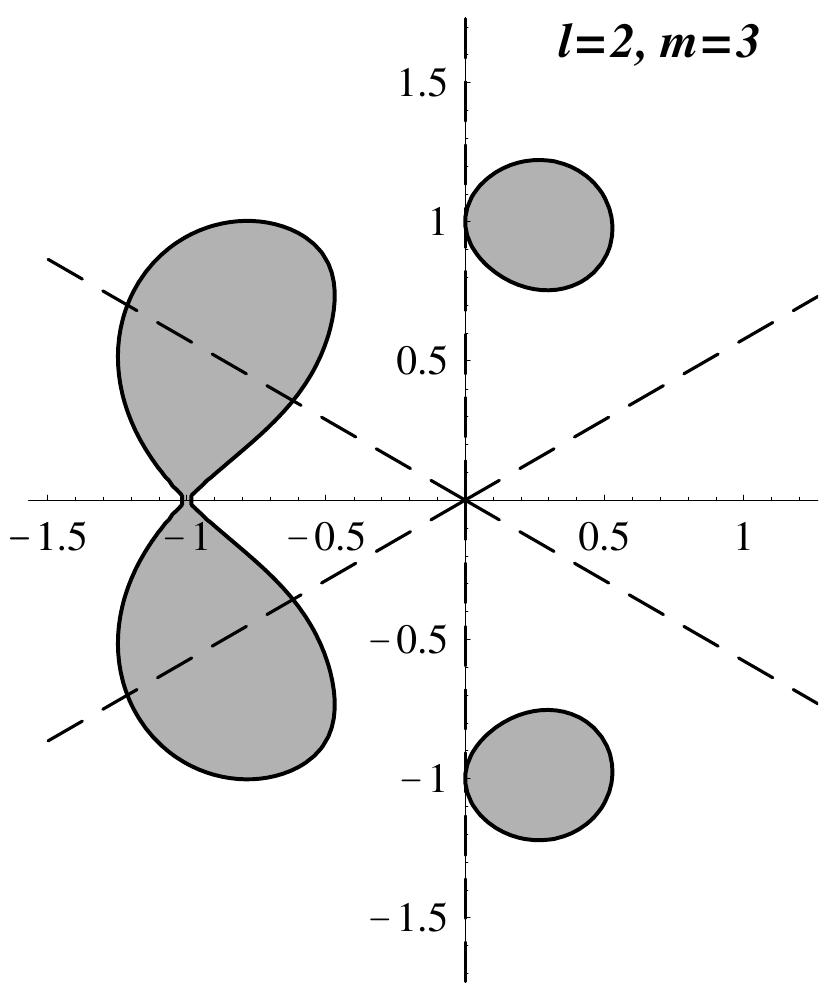}\hfill\plot{35mm}{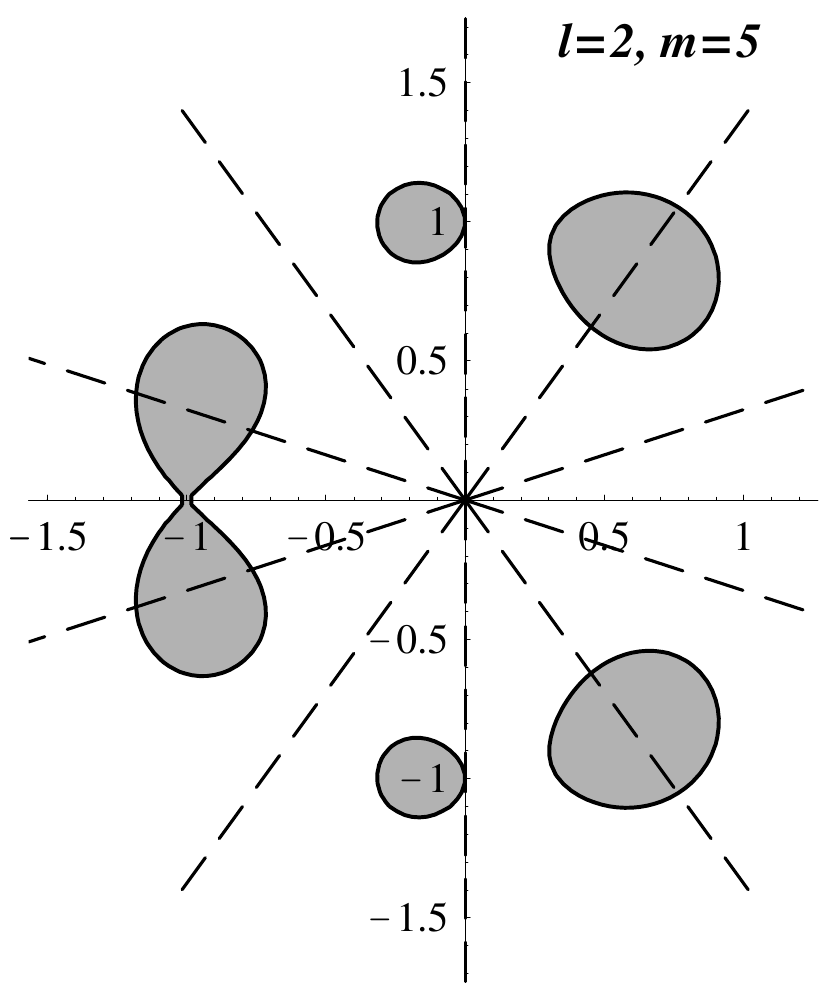}\hfill\plot{35mm}{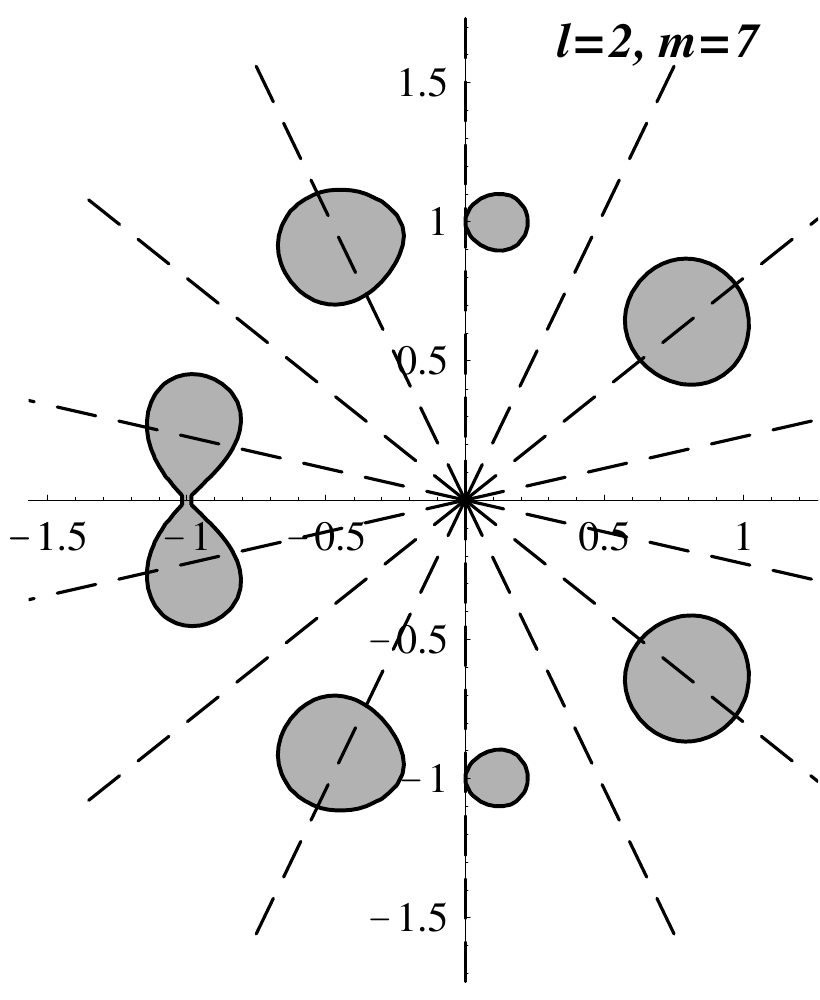}\\
\caption{The values of $q=\rho e^{\iu\varphi}$ inside the bounded domains in
$\mathbb C$ correspond to the regular potentials $u(n,t)$. The values along the
dashed lines correspond to the potentials periodic in $t$.}
\label{fig:zones}
\end{figure}

The substitution of the 2-soliton Ansatz
\begin{equation}\label{2sol}
 w(n,t)=1+e_1+e_2+A_{12}e_1e_2,\quad e_i=q^n_i\exp(-\omega_it+\delta_i)
\end{equation}
into (\ref{wlm}) gives us the dispersion relation and the phase shift:
\begin{equation}\label{omegA}
 \omega_i=q^m_i-q^{-m}_i,\quad
 A_{ij}=\frac{(q^l_i-q^l_j)(q^m_i-q^m_j)}{(1-q^l_iq^l_j)(1-q^m_iq^m_j)}.
\end{equation}
The direct check proves that then the 3-soliton Ansatz
\[
 w=1+e_1+e_2+e_3+A_{12}e_1e_2+A_{13}e_1e_3+A_{23}e_2e_3+A_{12}A_{13}A_{23}e_1e_2e_3
\]
satisfies (\ref{wlm}) automatically. It is interesting to compare these
formulas with their counterparts for the continuous SK equation
\cite{Sawada_Kotera, Caudrey_Dodd_Gibbon, DJKM, Parker}
\[
 e_i=\exp(\varkappa_ix-\omega_it+\delta_i),\quad
 \omega_i=\varkappa^5_i,\quad
 A_{ij}=
 \frac{(\varkappa_i-\varkappa_j)^2
       (\varkappa^2_i-\varkappa_i\varkappa_j+\varkappa^2_j)}
      {(\varkappa_i+\varkappa_j)^2
       (\varkappa^2_i+\varkappa_i\varkappa_j+\varkappa^2_j)}.
\]
Formula (\ref{2sol}) allows us to obtain the breather-type solutions as well,
if we choose
\[
  q_1=\rho e^{\iu\varphi},\quad
  q_2=\rho e^{-\iu\varphi},\quad
  \delta_1=\alpha+\iu\beta,\quad
  \delta_2=\alpha-\iu\beta.
\]
The regularity of the potential $u(n,t)$ is achieved under certain restrictions
on the value of $q$. In order to show this, rewrite relations (\ref{omegA}) as
follows:
\begin{gather*}
 \omega=\mu+\iu\nu,\quad
 \mu=(\rho^m-\rho^{-m})\cos m\varphi,\quad
 \nu=(\rho^m+\rho^{-m})\sin m\varphi,\\
 A_{12}=-\frac{4\rho^{m+l}\sin l\varphi\sin m\varphi}
              {(1-\rho^{2l})(1-\rho^{2m})},
\end{gather*}
then a simple algebra brings (\ref{2sol}) to the form
\[
 w=1+2z\cos y+A_{12}z^2,\quad
 y=\varphi n-\nu t+\beta,\quad z=\rho^ne^{\alpha-\mu t}.
\]
In particular, if $\varphi=\frac{2k+1}{2m}\pi$ then $\mu=0$ and solution $u$ is
periodic in $t$. The necessary and sufficient condition for $u$ to be regular
is that the function $w$ does not vanish at any $n,t$. In the generic case the
variables $y,z$ are independent and then this is equivalent to the condition
that the trinomial $1+2z+A_{12}z^2$ does not vanish at real $z$, that is
\[
 (\rho^l-\rho^{-l})(\rho^m-\rho^{-m})+4\sin l\varphi\sin m\varphi<0.
\]
Thus, we see that already two-phase solutions in these models exhibit a
nontrivial zone structure of the spectrum. The corresponding domains in the
plane $q=\rho e^{\iu\varphi}$ are shown on fig. \ref{fig:zones}, and the
examples of solutions $u(n,t)$ are shown on fig. \ref{fig:breather}.

\begin{figure}[t!]
\centerline{\plot{65mm}{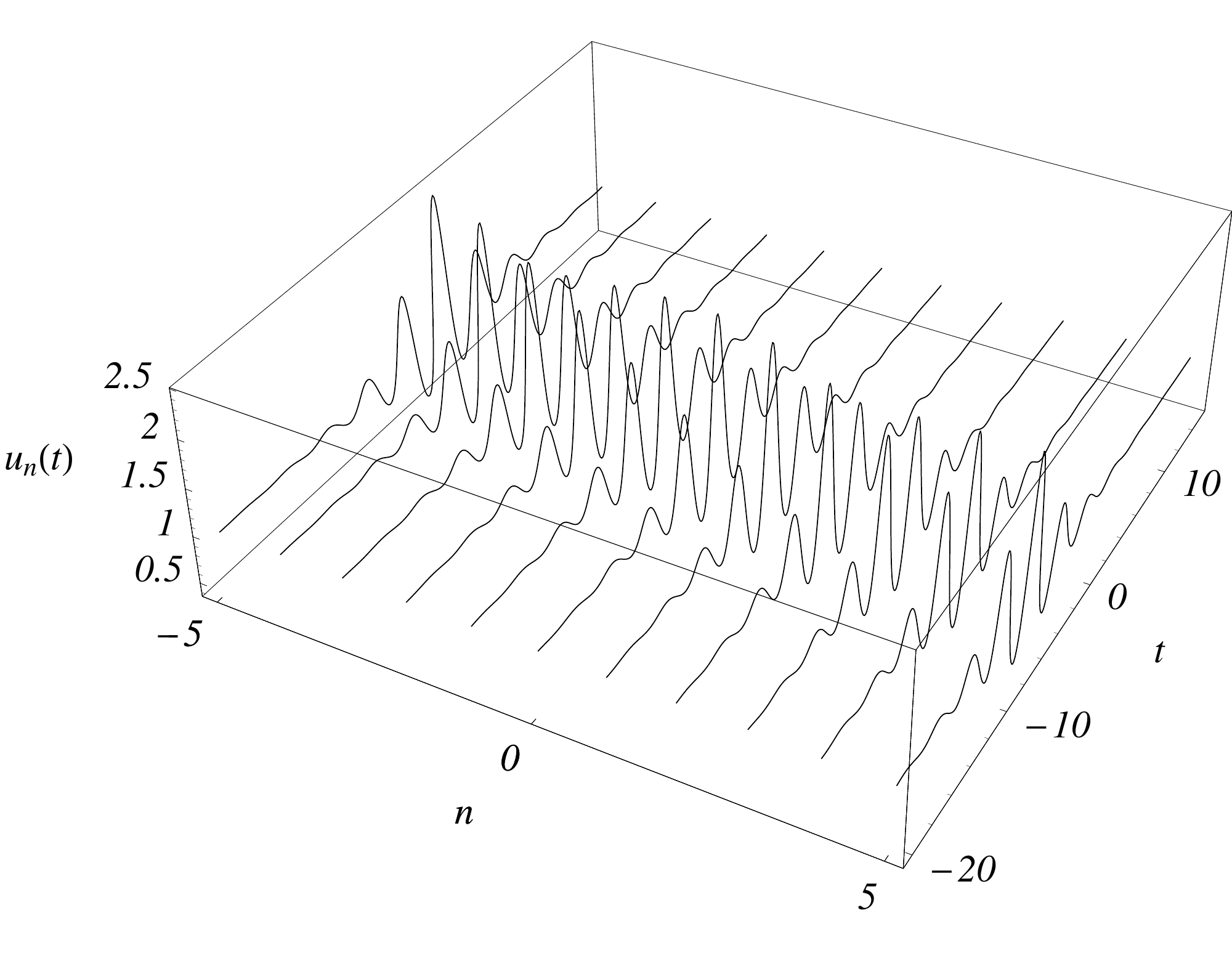}\hfill\plot{65mm}{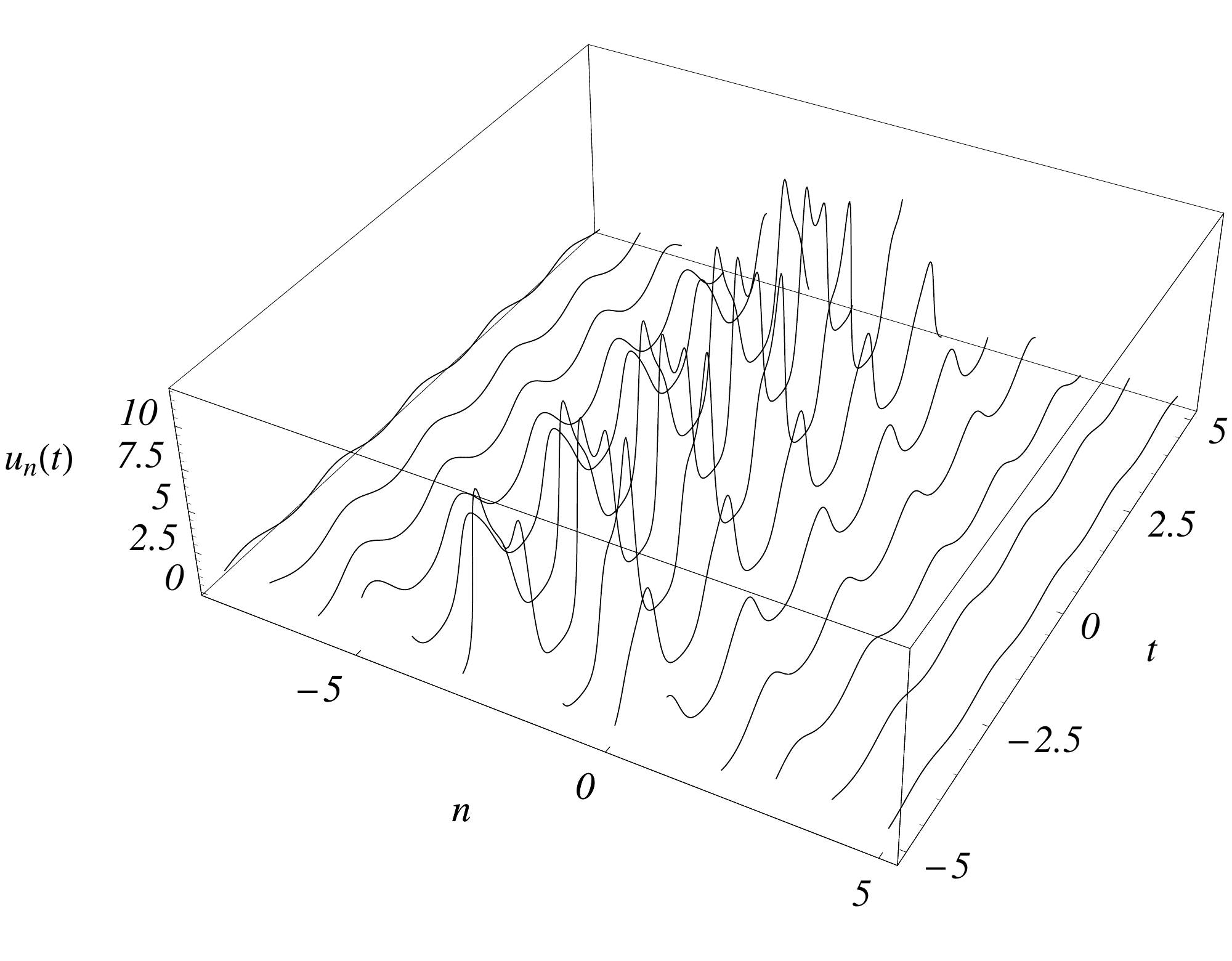}}
\caption{A moving and a stable breathers. The values of parameters:
$\rho=1.2$, $\varphi=2\pi/3$ (left); $\rho=1.6$, $\varphi=3\pi/4$ (right); in
both cases $l=1$, $m=2$, $\alpha=\beta=0$.}
\label{fig:breather}
\end{figure}

\section{A discrete analog of the Kaup--Kupershmidt equation}\label{s:dKK}

The Kaup--Kupershmidt equation
\[
 U_{,\tau}=U_5+5UU_3+\frac{25}{2}U_1U_2+5U^2U_1
\]
is associated with the spectral problem $L\psi=\lambda\psi$ where $L$ is the
skew-symmetric ordinary differential operator of third order
\[
 L=D^3+UD+\frac{1}{2}U_{,x}=(D-f)D(D+f),\quad U=2f_{,x}-f^2.
\]
When we find a discrete analog, a difficulty is that a symmetric or
skew-symmetric difference operator can be of even order only. A way to overcome
this is to consider a 6th order difference problem, but on the odd nodes of the
lattice only, so that effectively it is of 3rd order with respect to the double
shift $T^2$ (however, the coefficients may depend on the variables associated
with the even nodes as well). Let us consider the spectral problem
\begin{equation}\label{dKK.psi}
 u_{-3}\psi_{-3}+\psi_{-1}=\lambda(\psi_1+u\psi_3)
\end{equation}
or, in the operator form, denoting $K=uT^3+T$:
\[
 K^\dag\psi=\lambda K\psi.
\]
The Lax equation for the operator $L=K^{-1}K^\dag$ can be written in the form
of system (\ref{PQt}). It admits the reduction $B=-A^\dag$ which yields the
equation
\begin{equation}\label{KAKKA}
 K_{,t}+A^\dag K+KA=0.
\end{equation}
The operator $A$ is found as a Laurent polynomial with respect to the even
powers of $T$,
\[
 A=a^{(k)}T^{2k}+\dots+a^{(-k)}T^{-2k},
\]
and a direct analysis of equation (\ref{KAKKA}) at $k=1,2$ proves the following
statement.

\begin{theorem}
Equation (\ref{KAKKA}) with $K=uT^3+T$ is equivalent to the nonlocal lattice
equation
\[
 u_{,t_1}=u(f_2u_2-f_1u_1+f_{-1}u_{-1}-f_{-2}u_{-2})+f_1-f_{-1},\quad
 f_3u=f_{-1}u_2
\]
under the choice
\[
 A= -fT^2+f_{-2}u_{-2}-f_{-1}u_{-1}+f_{-3}T^{-2};
\]
and it is equivalent to the local lattice equation
\begin{equation}\label{dKK}
 u_{,t_2}=u(v_3-v_2+v_1-v_{-1}+v_{-2}-v_{-3}-u_2+u_{-2}),\quad v:=u_1uu_{-1}
\end{equation}
under the choice
\begin{gather*}
 A= u_1T^4-u_{-4}T^{-4}+(1-u_{-1}u_{-2})(T^2-T^{-2})\\
  +u_{-1}-u_{-2}-v+v_{-1}-v_{-2}+v_{-3}.
\end{gather*}
\end{theorem}

It is worth noticing that, alternatively, one can use the following pair of
operators (cf with the gauge equivalence (\ref{lm})):
\begin{equation}\label{dKK.psi'}
 \tilde K=uT^3+T^{-1},\quad
 \tilde A=-u_1u_{-1}T^4+u_{-2}u_{-4}T^{-4}-v+v_{-1}-v_{-2}+v_{-3}.
\end{equation}

The continuous limit to the KK equation is of the same general form as before,
namely, for the flow (\ref{dKK}) it reads
\[
 u(n,t_2)=\frac{1}{3}+\frac{4}{9}\varepsilon^2U\Bigl(x-\frac{8}{9}\varepsilon t_2,
  \tau+\frac{64\varepsilon^5}{135}t_2\Bigr),\quad x=\varepsilon n.
\]

\section{Examples related to generic operators}\label{s:generic}

Recall that, according to \cite{Blaszak_Marciniak}, the Bogoyavlensky type
lattices can be viewed as reductions of more general multi-field models
associated with the spectral problems $L\psi=\lambda\psi$ for generic
difference operators
$L=u^{(m)}T^m+u^{(m-1)}T^{m-1}+\dots+u^{(1-l)}T^{1-l}+u^{(-l)}T^{-l}$. Here
$m,l$ are any positive integers, and one can adopt the normalization
$u^{(m)}=1$ or $u^{(l)}=1$ without loss of generality. A part of the flows from
the corresponding hierarchy is consistent with the constraints
$u^{(m-1)}=\dots=u^{(1-l)}=0$ and this reduction brings to the Bogoyavlensky
lattices. A detailed study of some other reductions can be found in
\cite{Svinin}.

The lattices introduced in the previous sections are related with the spectral
problems $P\psi=\lambda Q\psi$ where operators $P,Q$ are binomial. It is
natural to expect that these lattices also define reductions for some
multi-field equations related with more general operators $P,Q$. The study of
such models is beyond the scope of the present paper and we restrict ourselves
by three typical examples.

\begin{example}
First, let us consider the Lax equations $P_{,t}=BP-PA$, $Q_{,t}=BQ-QA$ for the
binomial operators $P,Q$ with different potentials:
\[
 P=uT^3+T,\quad Q=T^2+v.
\]
If $v=u$ then operators $A,B$ are given by formulas (\ref{ABF}), (\ref{F}) with
a self-adjoint operator $F$ which contains only even powers of $T^2$. In the
general case two sets of operators $A,B$ appear, containing positive or
negative powers of $T^2$. The simplest operators and corresponding flows are
the following:
\begin{align*}
& A^- = v_{-2}v_{-1}T^{-2}+f_{-3}+f_{-2},\quad
  B^- = v_{-1}vT^{-2}+f_{-1}+f, \\
& \qquad
  \begin{aligned}
   u_{,t^-}&= u(f_{-1}-f_1),\\
   v_{,t^-}&= v(f+f_{-1}-f_{-2}-f_{-3}-v_1+v_{-1}),\quad f:=uv_1v_2;
  \end{aligned}\\[1ex]
& A^+ = u_{-2}u_{-1}T^2+g_{-1}+g,\quad
  B^+ = uu_1T^2+g+g_1, \\
& \qquad
  \begin{aligned}
   u_{,t^+}&= u(g+g_1-g_2-g_3-u_{-1}+u_1),\\
   v_{,t^+}&= v(g_1-g_{-1}),\qquad g:=u_{-2}u_{-1}v.
  \end{aligned}
\end{align*}
The flows $\partial_{t^-}$ and $\partial_{t^+}$ commute, and the flow
$\partial_{,t}=\partial_{t^-}-\partial_{t^+}$ admits the reduction $v=u$ which
brings to the dSK equation (\ref{dSK}). It should be noted that the same flows
can be obtained starting from the gauge equivalent operators $P=uT^3+T^2$,
$Q=T+v$.
\end{example}

\begin{example}
Now let us consider trinomial operators
\[
 P=uT^3+pT^2+T,\quad Q=T^2+qT+v.
\]
In this case operators $A,B$ contain the odd powers of $T$ as well. The
simplest operators and the corresponding flows are of the form
\begin{gather*}
 A^- = v_{-1}T^{-1}+v_{-1}p_{-2},\quad B^- = vT^{-1}+v_1p, \\
  \begin{aligned}
   u_{,x^-}&= u(u_{-1}q-u_1q_2-p+p_1),\\
   p_{,x^-}&= p(u_{-1}q-uq_1)+u-u_{-1},\\
   v_{,x^-}&= v(u_{-1}q-u_{-2}q_{-1}), \\
   q_{,x^-}&= uv_1-u_{-2}v;
  \end{aligned}
\end{gather*}
\begin{gather*}
 A^+ = u_{-2}T+u_{-2}q_{-1},\quad B^+ = uT+u_{-1}q, \\
  \begin{aligned}
   u_{,x^+}&= u(v_1p-v_2p_1),\\
   p_{,x^+}&= u_{-1}v-uv_2,\\
   v_{,x^+}&= v(v_1p-v_{-1}p_{-2}+q_{-1}-q), \\
   q_{,x^+}&= q(v_1p-vp_{-1})+v-v_1.
  \end{aligned}
\end{gather*}
\end{example}

\begin{example}
Let us consider the following generalization of the spectral problem
(\ref{dKK.psi'}):
\[
 K^\dag\psi=\lambda K\psi,\quad K=uT^3+v_{-1}T+T^{-1}.
\]
The isospectral deformations are defined by the operators
$A=a^{(k)}T^{2k}+a^{(k-1)}T^{2k-2}+\dots+a^{(-k)}T^{-2k}$. The simplest case
$k=1$ results in
\[
 A=u_{-1}T^2-u_{-2}T^{-2}+u_{-1}v_{-1}-u_{-2}v_{-2}
\]
and equation $K_{,t}+A^\dag K+KA=0$ is equivalent to the lattice
\begin{align*}
 u_{,t}&=-u(u_2v_2-u_1v_1+u_{-1}v_{-1}-u_{-2}v_{-2}-v_1+v_{-1}),\\
 v_{,t}&=-v(u_1v_1-u_{-1}v_{-1})+u_2u_1-u_{-1}u_{-2}+u_1-u_{-1}.
\end{align*}
The higher symmetry corresponding to $k=2$ is too bulky and we do not write it
down, however one can check that it admits the reduction $v=0$ to the dKK
equation (\ref{dKK}). In contrast, the flow $\partial_t$ itself does not admit
this reduction.
\end{example}

\section{Conclusion}

In this article we introduced a family of integrable lattice hierarchies
associated with fractional Lax operators. In particular, these hierarchies
contain equations found earlier in \cite{Tsujimoto_Hirota,Hu_Clarkson_Bullough}
by use of the Hirota bilinear formalism. We proved that these equations serve
as semi-discrete analogs of SK and KK equations. An important question which
remains open is about the Hamiltonian structure of the presented equations. As
usually, the existence of Lax representation allows to obtain a set of
conserved quantities which presumably are Hamiltonians, and moreover, the
applicability of $r$-matrix approach suggests that some more or less standard
Poisson bracket should exist. However, no explicit answer is found yet. Another
intriguing question is about possible relations with the models introduced in
\cite{Pugay,Hikami} within the theory of the lattice $W$ algebras.

\section*{Acknowledgements}

We are grateful to Ya.P. Pugay, Yu.B. Suris and A.K. Svinin for many
stimulating discussions. The research of V.A. was supported by grant
NSh--6501.2010.2.

\end{document}